\documentclass[a4paper,11pt]{article}
\author{Federico~Poloni\footnote{Dipartimento di Informatica, Universit\`a di Pisa. Largo Pontecorvo 3, 56127 Pisa, Italy. E-mail \texttt{fpoloni@di.unipi.it}} {} and Giacomo~Sbrana\footnote{Corresponding author Tel.: +33 232824673.}\footnote{Neoma Business School. 1, rue du Mar{\'e}chal Juin, 76130 Mont-Saint-Aignan, France. E-mail \texttt{gsb@neoma-bs.fr}}}

\title{A note on forecasting demand using the multivariate exponential smoothing framework}
\usepackage{graphicx}
\usepackage{amsmath}
\usepackage{amssymb}
\usepackage{amsthm}
\usepackage{mathtools}

\usepackage{natbib}

\usepackage{tikz}
\usetikzlibrary{arrows, shapes, matrix}

\newtheorem{Theorem}{Theorem}
\newtheorem{Remark}[Theorem]{Remark}
\newtheorem{Lemma}[Theorem]{Lemma}
\newtheorem{Proposition}[Theorem]{Proposition}
\newtheorem{Corollary}[Theorem]{Corollary}

\newcommand{\E}[1]{\mathbb{E}\left[ #1 \right]}
\DeclarePairedDelimiter{\abs}{\lvert}{\rvert}
\DeclarePairedDelimiter{\norm}{\lVert}{\rVert}

\usepackage{hyperref}
\begin{document}
\baselineskip=14pt
\date{}
\newcommand{\m}[1]{\begin{bmatrix}#1 \end{bmatrix}}

\maketitle
\begin{abstract}
Simple exponential smoothing is widely used in forecasting economic time series. This is because it is quick to compute and it generally delivers accurate forecasts. On the other hand, its multivariate version has received little attention due to the complications arising with the estimation. Indeed, standard multivariate maximum likelihood methods are affected by numerical convergence issues and bad complexity, growing with the dimensionality of the model. In this paper, we introduce a new estimation strategy for multivariate exponential smoothing, based on aggregating its observations into scalar models and estimating them. The original high-dimensional maximum likelihood problem is broken down into several univariate ones, which are easier to solve. Contrary to the multivariate maximum likelihood approach, the suggested algorithm does not suffer heavily from the dimensionality of the model. The method can be used for time series forecasting. In addition, simulation results show that our approach performs at least as well as a maximum likelihood estimator on the underlying VMA(1) representation, at least in our test problems. 

\end{abstract}

\section{Introduction}
Simple exponential smoothing represents an important benchmark model when forecasting the demand for goods and services. The most attractive feature of this model is its ease of computation. Unfortunately, the same feature does not hold for its multivariate version due to the complications arising with the estimation. Indeed, this represents an obstacle for practitioners, discouraging the employment of this model in empirical analysis.  This paper addresses this relevant issue by providing a feasible and accurate estimation method for a multivariate exponential smoothing model. 

We focus on the following state-space representation of an unrestricted multivariate simple exponential smoothing model (See \cite[Chapter 8]{HARV2})
\begin{equation}
\begin{array}{l}
y_{t}=\mu_{t}+\epsilon_{t},\\

\mu_{t}=\mu_{t-1}+\eta_{t},
\end{array}
\label{eq:MRW}
\end{equation}
with $y_t,\mu_t,\epsilon_t,\eta_t\in\mathbb{R}^{N}$. The noises $\eta_t$ and $\epsilon_t$ characterizing the system are independent and identically distributed with expected value equal to zero and
\begin{equation}
\operatorname{cov} \m{
\epsilon_{t}\\
\eta_{t}\\
}= 
\m{\Sigma_{\epsilon} &  0 \\
0 & \Sigma_{\eta}
}
\label{eq:Matrix1}
\end{equation}
where $\Sigma_{\epsilon}>0$, $\Sigma_{\eta}>0$ and $0$ are $N\times N$ matrices. The nonstationary system \eqref{eq:MRW} is known as the \emph{structural process}, and the covariances as in (\ref{eq:Matrix1}) are called \emph{structural parameters}. The system can be reparametrized as a first order integrated vector moving average process (i.e. integrated VMA(1)), the so-called \emph{reduced form}, using the Wold representation theorem
\begin{equation}
z_t := y_{t}-y_{t-1}=\eta_{t}+\epsilon_{t}-\epsilon_{t-1}=u_{t}-\Theta u_{t-1}, \quad \E{u_{t}u_{t}^T}=\Sigma_{u},
\label{eq:VMA_SS}
\end{equation}
for a suitable $\Theta, \Sigma>0\in\mathbb{R}^{N\times N}$, and an innovation process $u_t$ which is uncorrelated, but not in general independent. The parameters can be chosen so that \eqref{eq:VMA_SS} is \emph{invertible}, i.e., all the eigenvalues of $\Theta$ have modulus smaller than 1. This version can be recast in the more familiar exponentially weighted moving average form (EWMA)
\begin{equation}
\hat{y}_{t}=(I-\Theta)y_{t-1}+\Theta \hat{y}_{t-1},\;\;\;\;\;\;\;\;\text{for $t=1,2,\dots, T$},
\label{eq:MEWMA}
\end{equation}
where $\hat{y}_{t}$ denotes the forecast of $y_{t}$.

Our method is based on the relation between the original model and several scalar aggregates of the form $x_t=w_t^Tz_t$, for suitable $w\in\mathbb{R}^N$. It is a consequence of Wold's decomposition theorem \cite[Section~2.1.3]{Lut} that each of these models is a MA(1), i.e.,
\[
 x_t = v_t - \psi v_{t-1}, \quad \E{v_t^2}=\sigma.
\]
(Notice that we use here $\sigma$ to denote a variance, rather than a standard deviation, for notational consistency with $\Sigma$ for multivariate processes.)

Visually, we can represent the estimation procedure using Figure 1: \\
\begin{center}
\begin{tikzpicture}
\matrix (m) [matrix of nodes, column sep=2.5cm,row  sep=1cm, align=center, nodes={rectangle,draw, text width = 3cm,anchor=center} ]{
Observed data $z_t=y_t-y_{t-1}$               & parameters $\Theta,\Sigma_u$ of the MA    & autocovariances $\Gamma_0,\Gamma_1$ of the model\\
aggregates $w^Tz_t$      & parameters $\psi,\sigma$ of aggregate models             & autocovariances $\gamma_0,\gamma_1$ of aggregate models\\
};
\draw[-triangle 60] (m-1-1) edge node[right]{aggregate} (m-2-1);
\draw[-triangle 60] (m-2-1) edge node[below]{estimate (ML)} (m-2-2);
\draw[-triangle 60] (m-2-2) edge node[below]{compute} (m-2-3);
\draw[-triangle 60] (m-2-3) edge node[left]{compute} (m-1-3);
\draw[-triangle 60] (m-1-3) edge node[above]{compute} (m-1-2);
\end{tikzpicture}
\end{center}
Namely, we aggregate the model using several vectors $w$, estimate each of the univariate (scalar) models $x_t=w_tz_t$ as a MA(1), and then make some algebraic computations to derive the parameters $\Theta$ and $\Sigma_u$ from them.

In order to make these closed-form computations possible, we need to derive explicit, closed-form relations that allow us to
\begin{itemize}
 \item express the parameters $\Theta,\Sigma_u$ as a function of the autocovariances $\Gamma_0:=\E{z_tz_t^T}$ and $\Gamma_1:=\E{z_tz_{t-1}^T}$. This step is described in Section~\ref{sec:autocov}.
 \item express $\Gamma_0$ and $\Gamma_1$ as a function of the autocovariances of the aggregate models. This step, together with a description of the whole estimation procedure, is shown in Section~\ref{sec:META}.
\end{itemize}
Asymptotic consistency and normality of the resulting estimator are proved in Section~\ref{sec:asympt}.

In contrast, a maximum likelihood (ML) estimator of the VMA model \eqref{eq:VMA_SS} would follow directly the ``missing arrow'' on our diagram between the observed data $z_t$ and the parameters $\Theta,\Sigma_u$. What we do instead essentially trades off one $N$-dimensional maximum likelihood procedure for several univariate ones. Since the ML estimator is affected by numerical convergence issues and bad complexity, growing with the dimensionality of the model \citep{Kas11}, estimating many small models rather than a large one is computationally favorable. In Section~\ref{sec:numerical} we compare the performance of ML with those of our suggested estimator, called \emph{META (Moment Estimation Through Aggregation)}. Simulation results show that the suggested approach is not only very simple and fast but is also remarkably efficient having performance that is as good as that of the standard multivariate maximum likelihood approach. 

\section{Literature review} \label{sec:literature}
As noted by \cite{DegH06}, ``\emph{There has been remarkably little work in developing multivariate versions of the exponential smoothing methods for forecasting.}'' We argue that this is probably due to the difficulties of estimating parameters in large-dimensional system. Our framework, also known as exponentially weighted moving average (EWMA), has a long tradition in forecasting time series (\cite{Muth}). Moreover, the EWMA belongs to the more general exponential smoothing family (see \cite{Gar06} and \cite{Hol04}). Despite its simplicity, this family represents a valid candidate in forecasting demand (see for example \cite{DekDO04}, \cite{FliL95}, \cite{Fli99}, \cite{MakSH00}, \cite{MooHS12}, \cite{MooHS13}).

More recently, the multivariate version of the EWMA model has been considered as production planning framework when forecasting aggregate demand. To overcome the estimation difficulties in the multivariate case, two strategies have been suggested, the so-called \emph{top-down} and \emph{bottom-up} approaches (see for example \cite{Lut_aggregate}, \cite{WIDIARTAetal} and \cite{SS}). In detail, let $w^T$ be a fixed weight vector, and suppose that we are interested in forecasting the aggregated process $x_t=w^Tz_t$. For instance, $w^T=\m{1/N & 1/N & \dots & 1/N}$ means that we are interested in the arithmetic mean of the observed variables. Then,
\begin{itemize}
 \item The \emph{top-down} approach consists in constructing directly $x_t=w^T z_t$ and forecasting this aggregate series applying a (scalar) estimator to $\{x_t\}_{t=1,2,\dots,N}$. Note that this loses information about the original process, since we do not use the single components of $z_t$ but only an aggregate function of them; hence less accuracy is to be expected.
 \item The \emph{bottom-up} approach consists in applying a scalar estimator to each of the $N$ time series $\{(z_t)_1\}_{t=1,2,\dots,T}, \{(z_t)_2\}_{t=1,2,\dots,T}, \dots \{(z_t)_N\}_{t=1,2,\dots,T}$, and forecasting each of them individually to obtain a forecast of the aggregated process $x_t=w^Tz_t$. Again, this method ignores the interdependence among the variables.
\end{itemize}
There is a vast literature comparing top-down and bottom-up approaches when forecasting aggregate demand. Without attempting to survey all the contributions on this topic, we refer to \cite{FliL95}, \cite{Fli99}, \cite{WeaKS01}, \cite{DekDO04}, \cite{ZotKC05}, \cite{ZotK07}, \cite{WIDIARTAetal}, \cite{CheB10}. More recently \cite{MooHS12}, \cite{MooHS13} consider in details alternative forecasting methods, such as the simple exponential smoothing, for predicting the demand for spare parts in the South Korean Navy.

A necessary condition to compare top-down and bottom-up approaches is the knowledge of the parameters of the multivariate demand planning framework. Indeed, once data are available, practitioners are faced with the challenge to estimate the parameters of the system whose dimension might be large. Therefore, a relevant gap left by \cite{SS} is that they do not provide any indication on how to derive the parameters of the framework using the available data (i.e. $y_t$). Indeed, quoting their conclusions: \emph{``this paper contains useful results assuming full knowledge of the parameters of the multivariate exponential smoothing. We are aware that this represents an ideal situation since, in empirical analysis, practitioners do not have such information and misspecification issues do usually arise [...]"}. This note fills this important empirical gap by providing an efficient and fast estimation procedure for the exponentially weighted moving average model, based on the same aggregation techniques used in their paper.


\section{Closed-form results} \label{sec:autocov}
It is easy to see that the autocovariances of the process $z_t$, expressed with both parametrizations \eqref{eq:MRW}--\eqref{eq:Matrix1} and \eqref{eq:VMA_SS}, are given by
\begin{equation}\label{EWMAautocov}
\begin{aligned}
\Gamma_{0} &:= \E{z_{t} z_{t}^T}=\Sigma_{u}+\Theta\Sigma_{u}\Theta^T=\Sigma_{\eta}+2\Sigma_{\epsilon},\\
\Gamma_{1} &:= \E{z_{t}  z_{t-1}^T}=-\Theta\Sigma_{u}=-\Sigma_{\epsilon}.
\end{aligned}
\end{equation}
It is important in the following that $\Gamma_1$ is a symmetric matrix. Hence it is easy to derive the following relations that express the structural parameters as a function of the reduced ones:
\[
 \Sigma_{\epsilon} = \Theta\Sigma_{u}, \quad \Sigma_{\eta} = \Sigma_{u}+\Theta\Sigma_{u}\Theta^T - 2\Sigma_{\epsilon}.
\]
The inverse relationship, i.e., how to construct the reduced parameters $\Theta,\Sigma_u$ in terms of the structural ones $\Sigma_{\eta},\Sigma_{\epsilon}$, is less obvious; we present it in the following result.
\begin{Proposition} \label{prop:param}
Consider the model \eqref{eq:MRW}--\eqref{eq:Matrix1} and its reparametrization \eqref{eq:VMA_SS}, and let $Q:=\Sigma_{\eta}\Sigma_{\epsilon}^{-1}$. Then, the following relation holds
\begin{equation}
\begin{array}{l}
\Theta=\frac{1}{2}\left(Q+2I-\left(Q^2+4Q\right)^{\frac{1}{2}}\right),\\
\Sigma_{u}=\Theta^{-1}\Sigma_{\epsilon}.\\
\end{array}
\label{eq:risultato}
\end{equation}
The matrix square root in this expression is well-defined since $Q^2+4Q$ is diagonalizable with all positive eigenvalues.
\end{Proposition}
\begin{proof}
Note that $\Gamma_{0}$ can be expressed as
\begin{equation} \label{forGamma0}
\Gamma_{0}=\Sigma_{u}+\Gamma_{1}\Sigma_{u}^{-1}\Gamma_{1},
\end{equation}
since $\Gamma_{1}=-\Theta\Sigma_{u}=-\Sigma_{u}\Theta^T$. Post-multiplying \eqref{forGamma0} by $\Sigma_{u}^{-1}$, we have
\begin{displaymath}
-\Gamma_{0}\Gamma_{1}^{-1}\Theta = \Gamma_{0}\Sigma_{u}^{-1}=I+\Gamma_{1}\Sigma_{u}^{-1}\Gamma_{1}\Sigma_{u}^{-1}=I+\Theta^2.
\end{displaymath}
Therefore $\Theta$ satisfies the quadratic matrix equation
\begin{equation} \label{qme}
\Theta^2+\Gamma_{0}\Gamma_{1}^{-1}\Theta+I=0,
\end{equation}
with $\Gamma_{0}\Gamma_{1}^{-1}=(\Sigma_\eta+2\Sigma_\epsilon)(-\Sigma_{\epsilon})^{-1}= -Q-2I$. The matrix $Q$ is always diagonalizable with positive eigenvalues, since it is the product of two positive-definite matrices \cite[Theorem~7.6.3]{HJ}. Hence we can set $Q=PDP^{-1}$, $D=\operatorname{diag}(d_1,d_2,\dots,d_N)$, with $d_i>0$ for each $i=1,2,\dots,N$. Pre- and post-multiplying \eqref{qme} by $P^{-1}$ and $P$, we get
\[
 \tilde{\Theta}^2 -(D+2I) \tilde{\Theta} + I = 0, \quad \tilde{\Theta}:= P^{-1}\Theta P.
\]
The solutions of this matrix equation are given by diagonal matrices $\tilde{\Theta}=\operatorname{diag}(g_{1},g_{2},\dots,g_{N})$, where
\[
 g_i^2 - (d_i+2)g_i + 1=0, \quad i=1,2,\dots,N.
\]
The usual formula for the quadratic solution gives
\begin{equation} \label{quadratic}
g_{i}=\frac{d_{i}+2\pm\sqrt{d_{i}^{2}+4d_{i}}}{2}.
\end{equation}
Note that $g_i$ are the eigenvalues of $\tilde{\Theta}$, and hence of $\Theta$. Since $d_i < \sqrt{d_i^2+4d_i} < d_i+2$ whenever $d_i>0$, we have $0<g_i<1$ if we choose the minus sign and $g_i>1$ if we choose the plus sign. Hence we choose the minus sign to obtain invertibility of the resulting system \eqref{eq:VMA_SS}.

Putting back together the matrices, we get
\[
 \Theta = P\frac{1}{2}\left(D+2I - (D^2+4D)^{1/2} \right)P^{-1} = \frac{1}{2}\left(Q+2I - (Q^2+4Q)^{1/2} \right).
\]
The matrix square root is well defined since $Q^2+4Q$ has positive eigenvalues $d_i^2+4d_i$, $i=1,2,\dots,N$.

Finally, the second equation in \eqref{eq:risultato} follows from the second one in \eqref{EWMAautocov}, since we have already observed that $g_i<0$ and thus $\Theta$ is nonsingular.
\end{proof}
\begin{Remark}
The results as in (\ref{eq:risultato}) are the multivariate extension of the univariate results (see for example \citep{Muth} and \citep[p.~68]{HARV2}) with $Q$ representing a ``signal to noise'' matrix ratio. In general, $Q$ is not a symmetric matrix and therefore neither $\Theta$ is.
\end{Remark}
The expression for $\Theta$ in \eqref{eq:risultato} is useful for forecasting the system \eqref{eq:MEWMA}. Indeed, using the lag operator $L$ (such that $L y_t=y_{t-1}$) we can write the optimal linear forecasting for  \eqref{eq:MEWMA} as
\begin{equation} \label{forecaster}
y^F_{t+1}=(I- \Theta)(I- \Theta L)^{-1}y_{t}=(I- \Theta)\sum_{j=0}^{\infty}\Theta^{j}y_{t-j}.
\end{equation}

\begin{Corollary}
The reduced form parameters can be expressed in terms of the autocorrelations of $z_t$ as
\begin{equation} \label{redasautocov}
\begin{aligned}
\Theta &= -\frac{1}{2}\left(\Gamma_{0}\Gamma_{1}^{-1}+\left(\Gamma_{0}\Gamma_{1}^{-1}\Gamma_{0}\Gamma_{1}^{-1}-4 I\right)^{\frac{1}{2}}\right),\\
\Sigma_u &=2\left(\Gamma_{0}\Gamma_{1}^{-1}+\left(\Gamma_{0}\Gamma_{1}^{-1}\Gamma_{0}\Gamma_{1}^{-1}-4 I\right)^{\frac{1}{2}}\right)^{-1}\Gamma_1.
\end{aligned}
 \end{equation}
\end{Corollary}
Therefore, $\Gamma_{0}$, $\Gamma_{1}$ is the only information needed to obtain $\Theta$ and $\Sigma_{u}$. The reader might be tempted to use this result as an estimator, computing sample covariances $\hat{\Gamma}_{0}=\frac{1}{T}\sum_{t=1}^T z_{t}z_{t}^T$, $\hat{\Gamma}_{1}=\frac{1}{T}\sum_{t=1}^{T-1} z_{t}z_{t-1}^T$ and substituting them into \eqref{redasautocov}. In empirical analysis, however, these sample covariances might not be accurate enough. To solve this issue, in the next section we provide a method to derive these moments more accurately.

\section{Moment estimation through aggregation (\emph{META})} \label{sec:META}

Consider a generic multivariate MA(1) process
\begin{equation} \label{eq:VMAz}
  z_t = u_t - \Theta u_{t-1}, \quad \E{u_t u_t^T}=\Sigma_u, 
\end{equation}
and define its autocovariance matrices $\Gamma_k:=\E{z_t z_{t-k}^T}$; due to the structure of the process, $\Gamma_k=0$ for $\abs{k}>1$, and $\Gamma_0=\Gamma_0^T$.

We are interested in aggregate processes, that is, scalar processes of the form $x_t := w^T z_t$, for some vector $w\in\mathbb{R}^n$. This form includes in particular the components $(z_t)_1,(z_t)_2,\dots,(z_t)_N$ of the vector process $z_t$, which are obtained by setting $w=e_i$, for $j=1,2,\dots,N$, where $e_i$ is the $i$-th vector of the canonical basis, that is, the $i$-th column of $I_N$.

It turns out that if $\Gamma_1=\Gamma_1^T$ (as is the case in our EWMA setting, due to \eqref{EWMAautocov}), then we can recover these covariances by knowing those of some special aggregate processes.
\begin{Lemma} \label{gammatogamma}
 Let $z_t$ be a VMA(1) process \eqref{eq:VMAz}, and suppose that $\Gamma_1=\Gamma_1^T$. Given a vector $w\in\mathbb{R}^N$, $w\neq 0$, define the aggregate $x^{(w)}_t:=w^Tz_t$, and let $\gamma^{(w)}_k$ be its covariances.
 Then, the entries of $\Gamma_k$ are given by
 \begin{equation}  \label{GammaFromGammini}
  (\Gamma_k)_{i,j}=\begin{cases}
                    \gamma^{(e_i)}_k & i=j,\\
                    \frac12 \left(\gamma^{(e_i+e_j)}_k - \gamma^{(e_i)}_k - \gamma^{(e_j)}_k \right) & i \neq j.
                   \end{cases}
 \end{equation}
In particular, they are uniquely determined given the covariances of the $\frac{N(N+1)}{2}$ scalar processes constructed with vectors $w\in\mathcal{W}$,
\[
 \mathcal{W} :=  \{e_i \colon 1\leq i \leq N \} \cup \{e_i+e_j \colon 1 \leq i < j \leq N \}.
\]
\end{Lemma}
\begin{proof}
 Note that $\gamma^{(w)}_k=\E{w^T z_t z_{t-k}^T w} = w^T\Gamma_k w$. Hence, $\gamma^{(e_i)}_k = (\Gamma_{k})_{ii}$, and $\gamma^{(e_i+e_j)}_k = (\Gamma_k)_{ii} + (\Gamma_k)_{ij} + (\Gamma_k)_{ji}+ (\Gamma_k)_{jj} = (\Gamma_k)_{ii} + 2(\Gamma_k)_{ij} +(\Gamma_k)_{jj}$.
\end{proof}
Each aggregate process can be reparametrized as a scalar MA(1) itself (see \citep{Lut_aggregate}); hence, one can write
\begin{equation} \label{MAini}
 x^{(w)}_t = v^{(w)}_t - \psi^{(w)} v^{(w)}_{t-1}, \quad \E{(v^{(w)}_t)^2} = \sigma^{(w)},
\end{equation}
for suitable white noise sequences $v^{(w)}_t$. Note that, although each $v^{(w)}_t$ is a white noise sequence on its own, two generic entries $v^{(w_1)}_{t_1}$ and $v^{(w_2)}_{t_2}$, for given $t_1,t_2$ and $w_1\neq w_2$, might be correlated.

One can use this representation to express the autocovariances as a function of these parameters:
\begin{equation}\label{autocovFromMAini}
 \begin{aligned}
  \gamma^{(w)}_0 &= (1+(\psi^{(w)})^2)\sigma^{(w)},\\
  \gamma^{(w)}_1 &= -\psi^{(w)}\sigma^{(w)}.
 \end{aligned}
\end{equation}

This approach suggests an estimation procedure as follows. Given $T$ observations 
of the process \eqref{eq:VMAz}:
\begin{enumerate}
 \item For each of the $N(N+1)/2$ vectors $w\in \mathcal{W}$, construct the aggregate data $x_t^{(w)} = w^T z_t$, and estimate the MA(1) model \eqref{MAini}, obtaining $\hat{\psi}^{(w)}$ and $\hat{\sigma}^(w)$.
 \item For each $w$, construct $\hat{\gamma}_0^{(w)}$ and $\hat{\gamma}_1^{(w)}$ using the formulas \eqref{autocovFromMAini}.
 \item Recover estimates $\hat{\Gamma}_0$ and $\hat{\Gamma}_1$ using \eqref{GammaFromGammini}.
 \item Recover estimates $\hat{\Theta}$ and $\hat{\Sigma}_u$ using \eqref{redasautocov}.
\end{enumerate}
The advantage of steps 1--3 of this procedure with respect to an estimator based on the sample moments $\frac1T \sum_{t=1}^T z_t z_{t-k}^T$ is that a maximum likelihood estimator as above yields more accurate values for the asymptotic moments.

For the sake of simplicity, in order to provide intuition to the reader, we give an example using a bivariate model.
Consider the following system with two variables
\begin{align*}
x_{1t}&=u_{1t}+\phi_{11}u_{1t-1}+\phi_{12}u_{2t-1},\\
x_{2t}&=u_{2t}+\phi_{21}u_{1t-1}+\phi_{22}u_{2t-1}.
\end{align*}
The previous model can be reparametrized equation-by-equation as
\begin{align*}
x_{1t}&=\upsilon_{1t}+\psi_1\upsilon_{1t-1},\\
x_{2t}&=\upsilon_{2t}+\psi_2\upsilon_{2t-1}
\end{align*}
with $\E{\upsilon_{it}^2}=\sigma_i$. Finally consider the MA(1) process derived from the simple aggregation of the two components
\begin{displaymath}
\begin{array}{c}
x_t=x_{1t}+x_{2t}=a_{t}+\alpha a_{1t-1}\\
\end{array}
\end{displaymath}
with $\E{a_{t}^2}=\sigma_a$. Using the results above, we can now rewrite $\Gamma_0$ and $\Gamma_1$ as function of the parameters of the aggregate models as follows
\begin{align*}
\Gamma_0 &=\m{
(1+\psi_1^2)\sigma_{1} & \frac{1}{2}[(1+\alpha^2)\sigma_{a}-(1+\psi_1^2)\sigma_{1}-(1+\psi_2^2)\sigma_{2}]  \\
\frac{1}{2}[(1+\alpha^2)\sigma_{a}-(1+\psi_1^2)\sigma_{1}-(1+\psi_2^2)\sigma_{2}]  & (1+\psi_2^2)\sigma_{2}\\
},\\
\Gamma_1 &=\m{
\psi_1\sigma_{1} & \frac{1}{2}[\alpha\sigma_{a}-\psi_1\sigma_{1}-\psi_2\sigma_{2}]  \\ \\
\frac{1}{2}[\alpha\sigma_{a}-\psi_1\sigma_{1}-\psi_2\sigma_{2}]  & \psi_2\sigma_{2}\\
}.
\end{align*}

\section{Asymptotic properties} \label{sec:asympt}
As a first result, we prove that the aggregate MA processes that we estimate are well-behaved.
\begin{Lemma} \label{lemma:invini}
 Suppose that the process \eqref{eq:VMAz} is invertible. Then, for each $w\in\mathbb{R}^N$ with $w\neq 0$, the process \eqref{MAini} is invertible, and $\sigma^{(w)}>0$.
\end{Lemma}
\begin{proof}
Invertibility of \eqref{eq:VMAz} means that its autocovariance generating function \cite[\textsection 3.5]{BD} $\Gamma(z)$ is nonsingular for each $z$ on the unit circle, i.e.,
\[
 \Gamma(z) > 0 \quad \text{if $\abs{z}=1$}.
\]
The autocovariance generating function of \eqref{MAini} is 
\[
 (1-z \psi^{(w)})\sigma^{(w)} (1-z^{-1}\psi^{(w)})  = \gamma^{(w)}(z) = w^T \Gamma(z) w.
\]
Since $\Gamma(z)$ is a positive-definite matrix for $\abs{z}=1$, we also have that $w^T \Gamma(z) w > 0$. Hence $(1-z \psi^{(w)})\sigma^{(w)} (1-z^{-1}\psi^{(w)})>0$ whenever $\abs{z}=1$, and this implies that $\sigma^{(w)}>0$ and that there is an invertible representation with $\abs{\psi^{(w)}}<1$ for \eqref{MAini}.
\end{proof}
Therefore in the following we assume without further mention that $\abs{\psi^{(w)}}<1$.

The (quasi)-maximum likelihood estimator on the representation \eqref{MAini}, using zero initial values for simplicity, is given by (dropping the $\cdot ^{(w)}$ superscript for ease of notation)
\[
 (\hat{\psi},\hat{\sigma}) := \arg \min_{(\tilde{\psi},\tilde{\sigma})} \sum_{t=1}^T \ell_t(\tilde{\psi},\tilde{\sigma}),
\]
with the unconditional negative log-likelihood function $\ell_t$ given for each pair of reals $\tilde{\psi},\tilde{\sigma}$ by
\begin{equation} \label{ll}
\ell_t(\tilde{\psi},\tilde{\sigma}):= \frac12 \log \tilde{\sigma} + \frac{\tilde{v}_t^2}{2\tilde{\sigma}}, \quad \tilde{v}_t := \sum_{k=0}^{t-1} \tilde{\psi}^k x_{t-k}.
\end{equation}
The $\tilde{v}_t$ satisfy the linear recurrence $\tilde{v}_1=x_1$, $\tilde{v}_t=x_t+\tilde{\psi} \tilde{v}_{t-1}$, and are a function of $\tilde{\psi}$ and of the observations. We set for brevity $\tilde{v}'_t:=\frac{\partial}{\partial\tilde{\psi}}\tilde{v}_t$. Notice that $\tilde{v}'_t$ is a linear function of $\tilde{v}_1,\dots,\tilde{v}_{t-1}$, as can be proved by induction using the relation
\[
 \tilde{v}'_t = \tilde{v}_{t-1} + \tilde{\psi} \tilde{v}'_{t-1}.
\]
Moreover, when $\tilde{\psi}=\psi$ (the correct value), then $\tilde{v}_t=v_t$. We first evaluate the Hessian of the likelihood at the exact system parameters $(\psi,\sigma)$: by ergodicity,
\begin{multline} \label{hessian}
 \frac1T \sum \nabla^2 \ell_t(\psi,\sigma) \to
 \E{\m{\frac{\partial^2}{\partial\tilde{\psi}^2} \ell_t(\psi,\sigma) & \frac{\partial^2}{\partial\tilde{\psi} \partial \tilde{\sigma}} \ell_t(\psi,\sigma) \\ \frac{\partial^2}{\partial\tilde{\psi} \partial \tilde{\sigma}} \ell_t(\psi,\sigma) & \frac{\partial^2}{\partial\tilde{\sigma}^2} \ell_t(\psi,\sigma) }}
 \\=\m{
  \E{\frac1\sigma \left( {v'_t}^2 + \frac{\partial^2 v_t}{\partial \psi^2} v_t \right)} &
  \E{-\frac1{\sigma^2} v'_t v_t}\\
  \E{-\frac1{\sigma^2} v'_t v_t} &
  \E{\frac1{2\sigma^2}\left(\frac{2v_t^2}{\sigma}-1\right)}
 } = \m{\frac{1}{1-\psi^2} & 0 \\ 0 & \frac{1}{2\sigma^2}}.
\end{multline}
In evaluating these expected values, we have used the following facts:
\begin{itemize}
 \item $\E{v_t^2}=\sigma$.
 \item $v'_t$ and $\frac{\partial^2 v_t}{\partial\tilde{\psi}^2}$ are uncorrelated from $v_t$, since they are linear functions of $v_1, v_2, \dots, v_{t-1}$.
 \item $\E{{v'_t}^2}=\frac{\sigma}{1-\psi^2}$. The simplest way to prove this relation is through
 \begin{multline*}
 \E{{v'_t}^2} = \E{\left(v_{t-1}+\psi v'_{t-1} \right)^2}
 =\E{v_{t-1}^2 + 2v_{t-1}\psi v'_{t-1} + \psi^2{v'_{t-1}}^2} \\= \sigma + 0 + \psi^2\E{{v'_{t-1}}^2},
 \end{multline*}
and by stationarity $\E{{v'_t}^2} = \E{{v'_{t-1}}^2}$.
\end{itemize}

We continue by proving that the estimates of the scalar parameters $\hat{\psi}^{(w)}, \hat{\sigma}^{(w)}$ are consistent and asymptotically normal. The consistency part is easier, since we can consider each aggregated process independently.
\begin{Lemma} \label{consistency1}
 Let the model \eqref{eq:VMAz} be stationary and ergodic, with $\Sigma_u>0$. Then, the maximum likelihood estimators $\hat{\psi}^{(w)}, \hat{\sigma}^{(w)}$ are asymptotically consistent.
\end{Lemma}
\begin{proof}
Let us consider the generic combination $x_t=w^T z_t$, $t=1,2,\dots,T$; as stated above, this is a MA(1) process with weak (uncorrelated but not independent) white noise $v_t$.

We make use the general results on ML consistency in \cite[Theorem 1(a)]{LinM10}. Since the Hessian \eqref{hessian} is asymptotically nonsingular, the maximizing point is isolated. We have
\[
 \tilde{v}_t=\sum_{i\geq 0} \tilde{\psi}^i x_{t-i} = \sum_{i\geq 0} \tilde{\psi}^i (v_{t-i}-\psi v_{t-i-1}) = v_t + \sum_{i\geq 0} \tilde{\psi}^i (\tilde{\psi}-\psi) v_{t-i},
\]
hence
\[
 \E{\tilde{v}_t^2} = \sigma \left(1+\sum_{i\geq 0} \tilde{\psi}^{2i} (\tilde{\psi}-\psi)^2\right).
\]
If we restrict the parameter set to a compact set with $\tilde{\psi}<1$, the sum converges and thus $\sup_{\tilde{\psi}}\E{\tilde{v}_t^2}<\infty$. Hence the hypotheses in \cite{LinM10} hold and each aggregated process is asymptotically consistent.
\end{proof}
Establishing asymptotic normality is more involved: since the $v^{(w)}_t$ are neither independent nor uncorrelated from each other, we cannot rely on the classical central limit results. We use instead a central limit result for weakly dependent sequences from \cite{PelU06}, which we summarize and report as follows.
\begin{Theorem} \label{th:PelU06}
 For an i.i.d. sequence of random variables $(Y_i)_{i\in\mathbb{Z}}$, denote by $\mathcal{F}_a^b$ the $\sigma$-field generated by $Y_t$ with $a\leq t\leq b$ and define $\xi_t=f(Y_t, Y_{t-1}, \dots)$, $t\in\mathbb{Z}$. Assume that $\E{\xi_0}=0$, $\E{\xi_0^2}<\infty$, and
\begin{equation} \label{PelUcond}
 \sum_{t=1}^{\infty}\frac{1}{\sqrt{t}}\norm{\xi_0-\E{\xi_0 \mid \mathcal{F}^0_{-t}}}_2<\infty.
\end{equation}
Then,
\begin{equation} \label{PelUresult}
 \frac{1}{\sqrt{T}}\sum_{t=1}^T \xi_t \to N(0,\E{\xi_0^2}).
\end{equation}
\end{Theorem}
Indeed, \cite{PelU06} contains a stronger result on triangular sequences (Corollary~5); our statement \eqref{PelUresult} is a special case that can be obtained by setting
\[
 a_i=\begin{cases}
      1 & i=0\\
      0 & \text{otherwise}
     \end{cases}
\]
in the thesis of their Theorem~1, so that $b_n=\sqrt{n}$.

In the process \eqref{eq:VMA_SS}, the i.i.d. variables are $Y_t=\m{\epsilon_t\\ \eta_t}$, and we aim to prove that each of the $\nabla \ell_t(\psi,\sigma)$ can be chosen as a $\xi_t$ that satisfies the above condition \eqref{PelUcond}. We start with a couple of lemmas.

\begin{Lemma} \label{pellemma1}
Consider the process \eqref{eq:VMA_SS}, with i.i.d. variables $Y_t=\m{\epsilon_t\\ \eta_t}$. For a fixed $w\in\mathbb{R}^N$, $w\neq 0$, define $v^{(w)}_t$ and $\psi^{(w)}$ as above. Then, there are vectors $g_k,h_k\in \mathbb{R}^{2N}$ such that
 \begin{align} \label{gdec}
  v^{(w)}_t &= \sum_{k=0}^{\infty} g_k Y_{t-k},\\
  {v'}^{(w)}_t &= \sum_{k=0}^{\infty} h_k Y_{t-k}, \label{hdec}
 \end{align}
with $\norm{g_k}=O((\psi^{(w)})^k)$ and $\norm{h_k}=O(k(\psi^{(w)})^k)$ for $k\to\infty$.
\end{Lemma}
\begin{proof}
 Let us drop the superscript $(w)$ for clarity. Using the lag operator $L$, one has
 \[
x_t = w^T z_t = w^T(\eta_t + (I-L)\epsilon_t) = \m{w^T & w^T}Y_t-L\m{0 & w^T}Y_t.
 \]
Hence 
\begin{multline*}
 v_t = (1-\psi L)^{-1} x_t = \sum_{k\geq 0} \psi^k L^k x_t = \sum_{k \geq 0} \psi^k L^k \m{w^T & w^T}Y_t - \sum_{k\geq 0} \psi^k L^{k+1}\m{0 & w^T}Y_t \\= \m{w^T & w^T} Y_0 + \sum_{k\geq 0} \psi^k L^{k+1} \m{\psi w^T & (\psi-1)w^T}Y^T.
\end{multline*}
Similarly, starting from
\[
 v'_t = \sum_{k \geq 1} k \psi^{k-1} L^k x_t,
\]
one gets the other result.
\end{proof}

\begin{Lemma} \label{pellemma2}
Consider the process \eqref{eq:VMA_SS}, with i.i.d. variables $Y_t=\m{\epsilon_t\\ \eta_t}$ with finite fourth moment. For a fixed $w\in\mathbb{R}^N$, $w\neq 0$, define $v^{(w)}_t$ and $\psi^{(w)}$ as above. Then, condition \eqref{PelUcond} holds for both $\xi_t={v'}^{(w)}_t v^{(w)}_t$ and $\xi_t=(v^{(w)}_t)^2-\sigma^{(w)}$.
\end{Lemma}
\begin{proof}
We may write
\begin{equation} \label{pqdec}
  v_0 = \underbrace{\sum_{k=0}^{t} g_k Y_{t-k}}_{:=p_t}+ \underbrace{\sum_{k>t} g_k Y_{t-k}}_{:=q_t},
\end{equation}
where $p_t$ is a function in the $\sigma$-field $\mathcal{F}^0_{-t}$ and $q_t$ is independent from it, and similarly
\begin{equation} \label{rsdec}
 v'_0 = \underbrace{\sum_{k=0}^{t} h_k Y_{t-k}}_{:=r_t}+ \underbrace{\sum_{k>t} h_k Y_{t-k}}_{:=s_t}.
\end{equation}
One has
\begin{multline*}
 \E{v'_0 v_0 \mid \mathcal{F}^0_{-t}} = \E{(p_t+q_t)(r_t+s_t)\mid \mathcal{F}^0_{-t}} \\= p_tr_t+\underbrace{\E{q_t \mid \mathcal{F}^0_{-t}}}_{=0} r_t + p_t \underbrace{\E{s_t \mid \mathcal{F}^0_{-t}}}_{=0} + \E{q_ts_t \mid \mathcal{F}^0_{-t}} = p_tr_t+ \E{q_ts_t \mid \mathcal{F}^0_{-t}},
\end{multline*}
thus
\begin{multline*}
 \norm*{v'_0 v_0 - \E{v'_0 v_0 \mid \mathcal{F}^0_{-t}}}_2 =
 \norm{q_tr_t+p_ts_t+q_ts_t + \E{q_ts_t \mid \mathcal{F}^0_{-t}}}_2 \\
 \leq \norm{q_t}_2 \norm{r_t}_2+ \norm{p_t}_2 \norm{s_t}_2+2\norm{q_t}_2 \norm{s_t}_2.
\end{multline*}
Since the decompositions \eqref{gdec}, \eqref{hdec}, \eqref{pqdec}, \eqref{rsdec} are into independent (orthogonal) components, one can estimate
\begin{align*}
 \norm{p_t}_2 &\leq \norm{v_0}_2, & \norm{q_t}_2 &=O\left( \sum_{k>t} \norm{g_k}\right) = O(\psi^t),\\
 \norm{r_t}_2 &\leq \norm*{\left(\frac{\partial}{\partial \psi} v_0\right)}, & \norm{s_t}_2 &= O\left( \sum_{k>t} \norm{h_k} \right) = O(t \psi^t).
\end{align*}
When estimating the two sums, we used the fact that $\sum_{t\geq 0} \psi^t = (1-\psi)^{-1}<\infty$ and $\sum_{t\geq 0} t \psi^t = \psi(1-\psi)^{-2}<\infty$. Putting everything together, we have proved that
\[
 \norm*{v'_0 v_0 - \E{v'_0 v_0 \mid \mathcal{F}^0_{-t}}}_2 = O(t \psi^t) \quad \text{for $t\to\infty$}.
\]
Hence the sum in \eqref{PelUcond} converges (indeed, even without the $\frac{1}{\sqrt{t}}$ term), and the condition is verified.

A similar reasoning works for $v_t^2-\sigma$: we have
\[
\norm*{v_t^2-\sigma - \E{v_t^2-\sigma \mid \mathcal{F}^0_{-t}}}_2 = \norm{2p_tq_t+ \E{q_t^2 \mid \mathcal{F}^0_{-t}}} = O(\psi^t).
\]
The fourth moment finiteness is needed in order to have $\E{\xi_0^2}<\infty$.
\end{proof}

We have now all the tools to prove the asymptotic normality of the aggregated system parameters.
\begin{Theorem} \label{normality1}
Consider the process \eqref{eq:VMA_SS}, with i.i.d variables $Y_t=\m{\epsilon_t\\\eta_t}$ with finite fourth moment. Suppose that the process is stationary and ergodic, and that $\Sigma_\epsilon>0$, $\Sigma_\eta>0$. Then, the maximum likelihood estimators $\hat{\psi}^{(w)}$, $\hat{\sigma}^{(w)}$, for all $w\in\mathcal{W}$, are jointly asymptotically normal.
\end{Theorem}
\begin{proof}
The first-order optimality conditions for the ML estimates state that $0=\frac1T \sum \nabla \ell_t(\hat{\psi},\hat{\sigma})$, where $\nabla$ denotes taking a gradient with respect to the pair of parameters $(\tilde{\psi},\tilde{\sigma})$. Using a multivariate Taylor expansion around $(\psi,\sigma)$, we get
\begin{equation} \label{taylor1}
 0=\frac1T \sum \nabla \ell_t(\psi,\sigma) + \left(\frac1T \sum \nabla^2 \ell_t(\tilde{\psi},\tilde{\sigma})\right)\left(\m{\hat{\psi}\\\hat{\sigma}}-\m{\psi\\\sigma}\right),
\end{equation}
for a suitable pair $(\tilde{\psi},\tilde{\sigma})$ lying in the segment that joins $(\hat{\psi},\hat{\sigma})$ and $(\psi,\sigma)$. If $(\hat{\psi},\hat{\sigma})$ are close enough to the exact values, then by continuity the Hessian matrix is invertible and bounded, thus we can rewrite \eqref{taylor1} as
\begin{equation} \label{taylor2}
 \m{\hat{\psi}\\\hat{\sigma}}-\m{\psi\\\sigma} = -\left(\frac1T \sum \nabla^2 \ell_t(\tilde{\psi},\tilde{\sigma})\right)^{-1} \left( \frac1T \sum \nabla \ell_t(\psi,\sigma)\right).
\end{equation}

This expansion \eqref{taylor2} is valid for every $w\in\mathcal{W}$ that we use as the aggregation weights.
Let $\beta$ be the vector obtained by stacking the vectors $\m{\psi^{(w_i)}\\ \sigma^{(w_i)}}$ for each $w_i\in\mathcal{W}$, one above the other, and $\hat{\beta}$ be similarly defined with their ML estimators. Stacking the Taylor expansions \eqref{taylor2} one above the other and multiplying by $\sqrt{T}$ we get
\begin{equation} \label{needsl}
 \sqrt{T}(\hat{\beta}-\beta) = -M(\tilde{\beta})^{-1}\frac{1}{\sqrt{T}}\sum_{t=1}^T
 \m{\frac{\partial}{\partial \psi^{(w_1)}} \ell_t^{(w_1)}\\
    \frac{\partial}{\partial \sigma^{(w_1)}} \ell_t^{(w_1)}\\
    \frac{\partial}{\partial \psi^{(w_2)}} \ell_t^{(w_2)}\\
    \frac{\partial}{\partial \sigma^{(w_2)}} \ell_t^{(w_2)}\\
    \vdots\\
    \frac{\partial}{\partial \sigma^{(w_{N(N+1)/2})}} \ell_t^{(w_{N(N+1)/2})}
   },
\end{equation}
with $M(\tilde{\beta})$ the block diagonal matrix containing $\frac1T \sum \nabla^2 \ell_t^{(w_i)}(\tilde{\psi^{(w_i)}},\tilde{\sigma^{(w_i)}})$ in its diagonal blocks. We know from the proof of Lemma~\ref{consistency1} and the consistency that $M(\tilde{\beta})$ converges almost surely to a diagonal matrix; if we prove that a central limit result holds for the vector sum in \eqref{needsl}, then the thesis follows by Slutsky's theorem. Theorem~\ref{th:PelU06} and Lemma~\ref{pellemma2} together give only a weaker result, i.e., that each component of the vector sum is asymptotically normal when considered alone. To prove \emph{joint} normality, we need a modification of the above proof. We shall prove that each linear combination of its entries is asymptotically normal, and use the Cramer-Wold device \cite[Proposition~6.3.1]{BD}. Let us take a generic linear combination
\[
 C_t:=\sum_{i=1}^{N(N+1)/2} a_i \frac{\partial}{\partial \psi^{(w_i)}} \ell_t^{(w_i)} + \sum_{i=1}^{N(N+1)/2} b_i \frac{\partial}{\partial \sigma^{(w_i)}} \ell_t^{(w_i)}.
\]
By the linearity of expectation and subadditivity of norms,
\begin{multline*}
 \norm*{C_0-\E{C_0 \mid \mathcal{F}^0_{-t}}}_2 \leq \sum_{i=1}^{N(N+1)/2} \abs{a_i} \norm*{\frac{\partial}{\partial \psi^{(w_i)}} \ell_t^{(w_i)} - \E{\frac{\partial}{\partial \psi^{(w_i)}} \ell_t^{(w_i)}  \mid \mathcal{F}^0_{-t}}}_2\\
 + \sum_{i=1}^{N(N+1)/2} \abs{b_i} \norm*{\frac{\partial}{\partial \sigma^{(w_i)}} \ell_t^{(w_i)} - \E{\frac{\partial}{\partial \sigma^{(w_i)}} \ell_t^{(w_i)}  \mid \mathcal{F}^0_{-t}}}_2,
\end{multline*}
and we know from the proof of Lemma~\ref{pellemma2} that each term in the right-hand side is $O(t(\psi^{(w_i)})^t)$. Setting $\psi_{\max}:=\max_{w\in \mathcal{W}} \abs{\psi^{(w)}}$, we have therefore
\[
 \norm*{C_0-\E{C_0 \mid \mathcal{F}^0_{-t}}}_2 = O(t \psi_{\max}^t),
\]
hence our generic linear combination $C_t$ satisfies the bound of Theorem~\ref{PelUresult} and thus is asymptotically normal. So, putting all together, in \eqref{needsl} $M(\tilde{\beta})^{-1}$ converges a.s.~to a constant matrix and the scaled sum converges in probability to a normal vector, thus by Slutsky's theorem $\sqrt{T}(\hat{\beta}-\beta)$ is asymptotically normal.
\end{proof}

Hence we have the following asymptotic result for $\hat{\Theta}$ and $\hat{\Sigma}_u$.
\begin{Theorem}
Consider the EWMA process \eqref{eq:VMA_SS}; suppose that the noises $\eta_t$ and $\epsilon_t$ are i.i.d. processes with variances $0<\Sigma_\eta,\Sigma_\epsilon<\infty$, and that the resulting EWMA process is stationary and ergodic. Then, the META estimator $\hat{\Theta}, \hat{\Sigma}_u$ described in Section~\ref{sec:META} is asymptotically consistent. If, in addition, the fourth moments of $\epsilon_t$ and $\eta_t$ are finite, then the estimator is asymptotically normal.
\end{Theorem}
\begin{proof}
The estimated values $\hat{\Theta}$ and $\hat{\Sigma}_u$ are an a function of $\psi^{(w)}$ and $\sigma^{(w)}$ for $w\in\mathcal{W}$ as introduced above; the specific function, obtained by composing \eqref{GammaFromGammini} and \eqref{eq:risultatoMoM}, is continuous and differentiable. Since it is continuous and these values are consistent by Lemma~\ref{consistency1}, the estimator is consistent. Under the additional hypothesis on the fourth moment, these values are also asymptotically normal by Theorem~\ref{normality1}, thus by the delta method \cite[Appendix~C.5]{Lut} asymptotic normality holds.
\end{proof}

\section{META vs. Maximum Likelihood: some numerical experiments} \label{sec:numerical}
In this section we provide some numerical experiments to compare the estimates obtained with the META method vis-a-vis a maximum likelihood estimator on the VMA(1) representation \eqref{eq:VMA_SS}.
We generated simulated data for a model of the form \eqref{eq:MRW}, using Gaussian noise with four different sets of covariance matrices, two bivariate and two trivariate ones, resulting in signal-to-noise ratios of different magnitude:
\begin{description}
 \item[Model 1] $\Sigma_\eta=\m{
1 &  -0.5 \\
-0.5 & 1.5}$, $
\Sigma_\epsilon=\m{
1.5 &  -0.15 \\
-0.15 & 1\\
} $,
\item[Model 2] $ \Sigma_\eta=\m{
1 &  -0.5 \\
-0.5 & 1.5}$, $
\Sigma_\epsilon=\m{
30 &  -3 \\
-3 & 20\\
}$,
\item[Model 3] $\Sigma_\eta=\m{
1 &  -0.5 & 0.3\\
-0.5 & 1.5 & -0.2\\
0.3 & -0.2 &1}$, $
\Sigma_\epsilon=\m{
1.5 &  -0.15 &-0.1\\
-0.15 & 1 & 0.3\\
-0.1 & 0.3 & 1.5}
$,
\item[Model 4] $\Sigma_\eta=\m{
1 &  -0.5 & 0.3\\
-0.5 & 1.5 & -0.2\\
0.3 & -0.2 &1}$, $
\Sigma_\epsilon=\m{
30 &  -3 &-2\\
-3 & 20 & 6\\
-2 & 6 & 30}
$.
\end{description}
The aim here is generating two different types of models; Model 1 and Model 3 have roots of the matrix $\Theta$ being about half of those of Model 2 and Model 4 respectively. Moreover, in Model 1 and Model 3 $\Sigma_u$ is much smaller than that of Model 2 and Model 4 respectively.

For each model, we generated time series of three different sample sizes $T=200$, $400$ and $1000$, and estimated them using both the ML and META methods. Each experiment has been repeated 500 times, with different data, produced each time using new computer-generated random numbers. All simulations were carried out using \textsc{Mathematica 8} by Wolfram and its \textsc{TimeSeries 1.4.1} package\footnote{Further information can be found in the Wolfram website: see http://media.wolfram.com/documents/TimeSeriesDocumentation.pdf }. The source files for the simulation are available upon request.

Since the parameter matrices for these simulated models are explicitly available (see Proposition~\ref{prop:param}), we can check how close the estimated $\Theta$ and $\Sigma_u$ are to the real ones. As error measure, we used the relative error in the Frobenius norm (root mean squared error of the matrix entries)
\begin{equation} \label{rmse}
RMSE=\frac{\left\|\hat{X}-X\right\|_F}{\left\|X\right\|_F},
\end{equation}
where $X$ is the matrix of true parameters and $\hat{X}$ is the estimated one, and $\left\|\cdot\right\|_F$ is the Frobenius norm: for a $m\times n$ matrix $X$, $\norm{X}_F:=\left(\sum_{i=1}^m \sum_{j=1}^n X_{ij}^2\right)^{1/2}$. The average errors on the set of 500 repeated experiments under this error metric are reported in Table~\ref{results}. The central columns contain the RMSE \eqref{rmse} multiplied by 1000. The last two columns report the average time (in seconds) taken by each estimation procedure.

Notice that the error is considerably lower in Models~2 and~4, which have a lower signal-to-noise ratio $Q$ and thus $\Theta$ with eigenvalues closer to the unit circle. This behaviour is expected in view of the properties of the maximum likelihood estimator for ARMA processes (cfr. \cite[Section~7.2.3]{Lut}, and \cite[\textsection 8.5]{BD} for a discussion of asymptotic efficiency in the scalar case). These empirical results show that META, whose derivation contains elements from both moment estimators and ML estimators, does not degrade in quality like the former when the eigenvalues are closer to the unit circle, but seems to maintain the higher asymptotic efficiency of the latter.

Overall, the results are clearly in favor of the META estimator. Indeed, not only the META estimator is extremely faster than the multivariate ML estimator, but it seems to outperform the rival estimator in terms of accuracy nearly all the times. There are only two cases where ML slightly outperforms META with respect to $\Sigma_u$ (i.e., Model 1 and Model 3 with $T=400$). On the contrary, regarding $\Theta$, META is always more accurate than the ML approach. The last column of Table 1 shows the computational difficulties of the standard ML approach. On a normal computer, it took us about one week of computation time to obtain the ML results for Models 3 and 4 with $T=1000$. In contrast, the same experiments with the META estimator took one hour and a half.
\begin{table}
\caption{Mean relative (normalized) RMSE of the estimates of the reduced parameters} \label{results}
{\centering
\begin{tabular}{|c|c|cc|cc|cc|} \hline
        &Sample size $T$&$\Theta$ META&$\Theta$ ML&$\Sigma_u$ META&$\Sigma_u$ ML &Time META&Time ML\\ \hline
		    &200  &202.52&236.77&108.28&109.11&1.55&59.99\\ 
Model 1	&400  &121.41&138.13&83.31&82.93&3.13&107.28\\
				&1000  &80.83&101.53&48.65&48.96&8.34&226.64\\ \hline
&200  &69.51&78.26&97.50&98.31&1.36&46.53\\
Model 2					&400  &48.26&56.48&80.91&81.52&3.74&113.36\\ 
				&1000  &28.01&34.07&47.60&48.02&7.63&210.06\\ \hline
		    &200  &205.07&254.49&135.26&136.41&2.65&201.41\\
Model 3	&400  &162.95&187.40&93.48&93.21&5.77&316.74\\
				&1000 &93.85&108.92&60.08&61.13&12.62&523.67\\ \hline
		    &200  &86.66&107.22&123.86&124.19&2.81&202.95\\
Model 4	&400  &57.03 &67.25 &95.13 &96.75 &5.66&295.93\\
				&1000 &29.91&37.04&61.78&62.13&13.12&541.36\\ \hline
\end{tabular}}
\\
\begin{scriptsize}
The cells relative to $\Theta$ and $\Sigma_u$ reports the average relative root mean squared error \eqref{rmse} of the estimated parameters, multiplied by 1000. The last two columns report the average number of seconds required for a single run of the estimation procedure. Lower is better.
\end{scriptsize}
\end{table}
These results clearly suggest the adoption of the suggested estimator, being not only extremely faster to compute, but also having very good small sample properties.

As a second performance test, for Models~1 and~2 and $T=200$, we compared the forecasting accuracy of the estimated parameters. In detail, we generated $T=200$ observations of each model, estimated the parameters using the first $199$ observations only, and used them to generate a prediction $y^F_{200}$ of the last value $y_{200}$ of the time series according to~\eqref{forecaster}. Each of these experiments has been repeated $R=200$ times, each time with a different randomly-generated time series of length $200$. We report in Figures~2 and~3 the boxplot of the forecast error $y_{200}-y^F_{200}$ for Model~1 and~2 respectively. Each figure contains box plots for the two components  $(y^F_{200}-y_{200})_1$ and $(y^F_{200}-y_{200})_2$; for example, \emph{META1}  and \emph{META2} refer respectively to the forecast error of the first and second component of the multivariate time series when META is employed. As a term of comparison, we report the prediction error obtained with the \emph{true} system matrix $\Theta$, which is available explicitly in our simulation. 

\begin{center}
\includegraphics{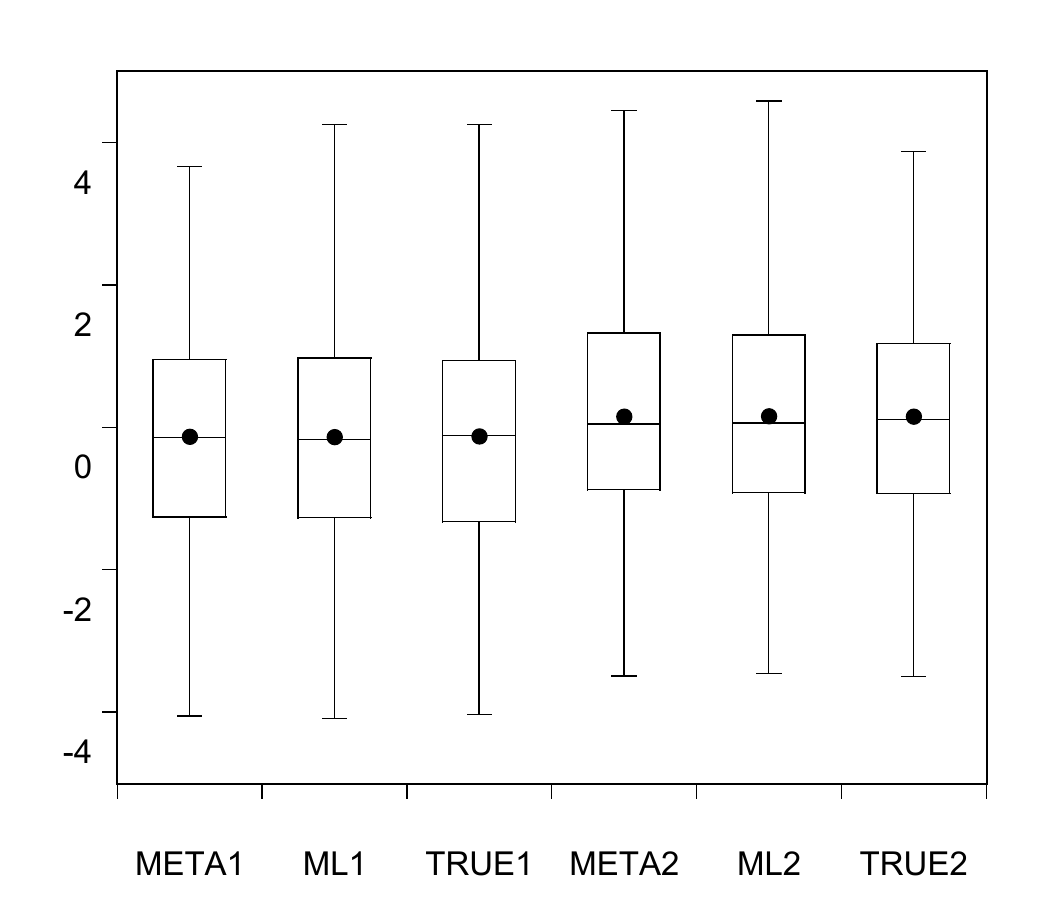}
\end{center}

\begin{center}
\includegraphics{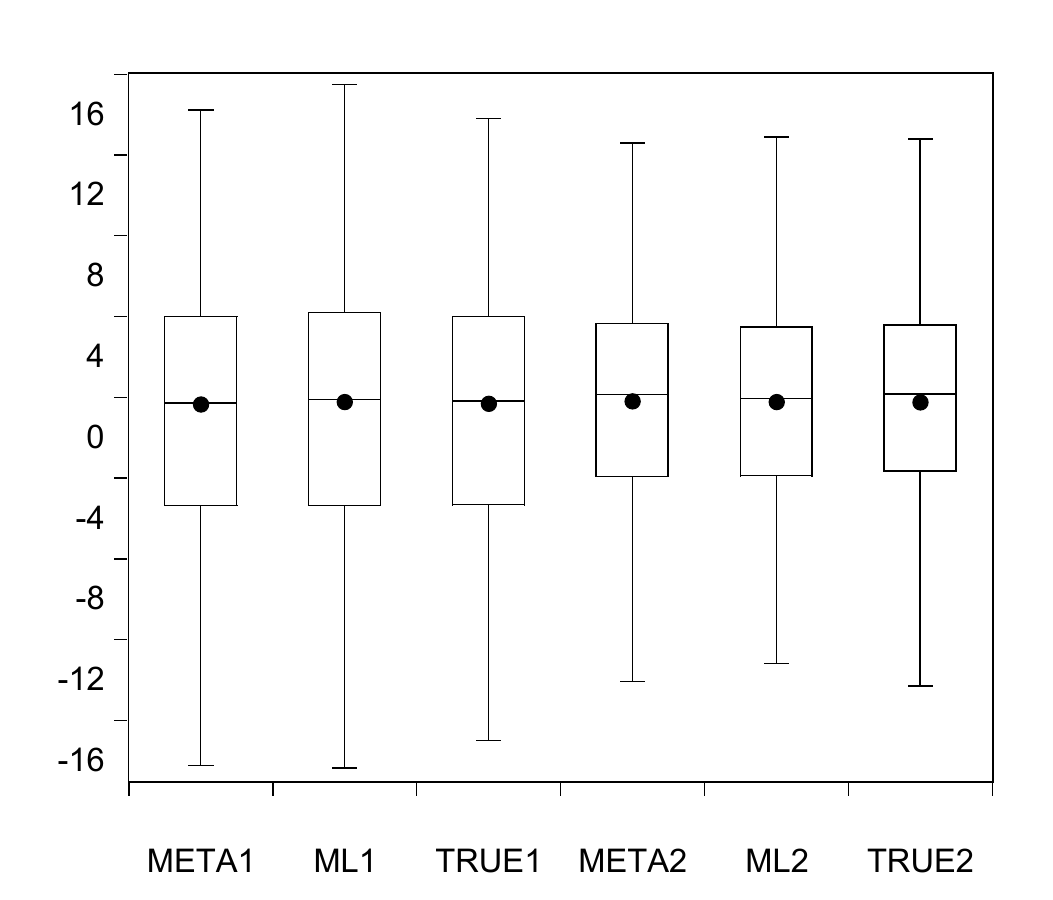}
\end{center}

As expected, the forecast errors relative to Model~2 are much more dispersed than those for Model~1. Moreover, the prediction accuracy is almost identical in all three cases for each equation relative to both Model~1 and~2, showing that the META forecasts are essentially as good as the ML ones (and almost as good as the real system matrices).

\section{Conclusions}
Simple exponential smoothing has been shown to be a valid candidate in forecasting demand (see \cite{DekDO04} and \cite{MooHS12}). This paper provides exact results linking reduced form parameters and autocovariances for the simple exponential smoothing in the multivariate framework. The results are used to provide a fast and efficient estimator, which seems to outperform the multivariate maximum likelihood on the underlying VMA representation in both time and accuracy. The technique used in the estimator allows one to reduce the problem from one $N$-dimensional maximum likelihood estimation to $N(N+1)/2$ scalar ML problems. This is especially convenient, since ML estimators for high-dimensional problems are slower to converge and more prone to numerical failures. The ML estimator has an expected complexity of $O(N^3T)$ per step, hence, under the reasonable assumption that the number of steps stays the same or decreases for the aggregate problems, passing to univariate problems rates to be even more effective when $N$ is larger. An additional benefit is that the scalar estimation problems are separate and can be solved in parallel.

A key feature of the VMA models resulting from our exponential smoothing model is that the autocovariance $\Gamma_1$ is a symmetric matrix. This property is used in both Proposition~\ref{prop:param} and Lemma~\ref{gammatogamma}. Our estimation procedure requires only this hypothesis, so it works without changes for any MA(1) model with $\Gamma_1=\Gamma_1^T$.
%
%
We are currently working on removing the assumption $\Gamma_1=\Gamma_1^T$. This more general model arises, for instance, when the noises $\epsilon_t$ and $\eta_t$ characterizing the system \eqref{eq:MRW} are correlated; we leave this for future research. Another open problem is deriving the exact asymptotic covariance of the estimator, with the aim of comparing it to maximum likelihood in terms of asymptotic efficiency (cfr. \cite[\textsection 8.5]{BD}). This task looks challenging, even in the case of Gaussian noise, since the noises $v^{(w)}$ of the aggregate processes are correlated and the computation would have to keep track of all these correlations.

A limitation of this work is that the suggested estimator is valid for the single exponential smoothing, but not for the whole family of exponential smoothing models. For example, our estimator cannot be used for models that take into account for the presence of a stochastic trend, such as the local linear trend model or the cubic smoothing spline models (see \cite{HARV2} and \cite{HynKPB05}). This is because the lack of closed-form results for more complex models.  

This manuscript has practical implications for practitioners involved in forecasting a multivariate production planning framework. Consider, for example, the case of a retail company providing a broad range of products to its customers. In order to reduce costs and to manage efficiently the production planning process, the company has to rely on accurate forecasts for the demand of each good/service as well as for the aggregate demand. Aggregated models, such as the top-down and bottom-up approaches, are often used because it is difficult and computationally intensive to handle a multivariate approach with full dependence between the variables. Instead, thanks to the algorithm suggested in this paper, estimation and forecasting can now be implemented with a multivariate exponential smoothing model without facing heavy computational issues.

\bibliographystyle{abbrvnat}
\bibliography{MEWMA_PolSbra}

\begin{thebibliography}{25}
\providecommand{\natexlab}[1]{#1}
\providecommand{\url}[1]{\texttt{#1}}
\expandafter\ifx\csname urlstyle\endcsname\relax
  \providecommand{\doi}[1]{doi: #1}\else
  \providecommand{\doi}{doi: \begingroup \urlstyle{rm}\Url}\fi

\bibitem[Brockwell and Davis(2006)]{BD}
P.~J. Brockwell and R.~A. Davis.
\newblock \emph{Time series: theory and methods}.
\newblock Springer Series in Statistics. Springer, New York, 2006.
\newblock ISBN 978-1-4419-0319-8; 1-4419-0319-8.
\newblock Reprint of the second (1991) edition.

\bibitem[Chen and Blue(2010)]{CheB10}
A.~Chen and J.~Blue.
\newblock Performance analysis of demand planning approaches for aggregating,
  forecasting and disaggregating interrelated demands.
\newblock \emph{International Journal of Production Economics}, 128\penalty0
  (2):\penalty0 586--602, Dec. 2010.
\newblock ISSN 0925-5273.
\newblock \doi{10.1016/j.ijpe.2010.07.006}.
\newblock URL
  \url{http://www.sciencedirect.com/science/article/pii/S0925527310002318}.

\bibitem[De~Gooijer and Hyndman(2006)]{DegH06}
J.~G. De~Gooijer and R.~J. Hyndman.
\newblock 25 years of time series forecasting.
\newblock \emph{International Journal of Forecasting}, 22\penalty0
  (3):\penalty0 443 -- 473, 2006.
\newblock ISSN 0169-2070.
\newblock \doi{http://dx.doi.org/10.1016/j.ijforecast.2006.01.001}.
\newblock URL
  \url{http://www.sciencedirect.com/science/article/pii/S0169207006000021}.
\newblock Twenty five years of forecasting.

\bibitem[Dekker et~al.(2004)Dekker, van Donselaar, and Ouwehand]{DekDO04}
M.~Dekker, K.~van Donselaar, and P.~Ouwehand.
\newblock How to use aggregation and combined forecasting to improve seasonal
  demand forecasts.
\newblock \emph{International Journal of Production Economics}, 90\penalty0
  (2):\penalty0 151--167, July 2004.
\newblock ISSN 0925-5273.
\newblock \doi{10.1016/j.ijpe.2004.02.004}.
\newblock URL
  \url{http://www.sciencedirect.com/science/article/pii/S0925527304000398}.

\bibitem[Fliedner and Lawrence(1995)]{FliL95}
E.~B. Fliedner and B.~Lawrence.
\newblock Forecasting system parent group formation: An empirical application
  of cluster analysis.
\newblock \emph{Journal of Operations Management}, 12\penalty0 (2):\penalty0
  119--130, 1995.

\bibitem[Fliedner(1999)]{Fli99}
G.~Fliedner.
\newblock An investigation of aggregate variable time series forecast
  strategies with specific subaggregate time series statistical correlation.
\newblock \emph{Computers \& Operations Research}, 26\penalty0 (10):\penalty0
  1133--1149, 1999.

\bibitem[Gardner~Jr.(2006)]{Gar06}
E.~Gardner~Jr.
\newblock Exponential smoothing: The state of the art -- part ii.
\newblock \emph{International Journal of Forecasting}, 22\penalty0
  (4):\penalty0 637--666, 2006.
\newblock \doi{10.1016/j.ijforecast.2006.03.005}.

\bibitem[Harvey(1991)]{HARV2}
A.~C. Harvey.
\newblock \emph{Forecasting Structural Time Series Models and the {K}alman
  Filter}.
\newblock Cambridge University Press, 1991.

\bibitem[Holt(2004)]{Hol04}
C.~C. Holt.
\newblock Author's retrospective on "forecasting seasonals and trends by
  exponentially weighted moving averages".
\newblock \emph{International Journal of Forecasting}, 20\penalty0
  (1):\penalty0 11--13, Jan. 2004.
\newblock ISSN 0169-2070.
\newblock \doi{10.1016/j.ijforecast.2003.09.017}.
\newblock URL
  \url{http://www.sciencedirect.com/science/article/pii/S0169207003001158}.

\bibitem[Horn and Johnson(1990)]{HJ}
R.~A. Horn and C.~R. Johnson.
\newblock \emph{Matrix analysis}.
\newblock Cambridge University Press, Cambridge, 1990.
\newblock ISBN 0-521-38632-2.
\newblock Corrected reprint of the 1985 original.

\bibitem[Hyndman et~al.(2005)Hyndman, King, Pitrun, and Billah]{HynKPB05}
R.~J. Hyndman, M.~L. King, I.~Pitrun, and B.~Billah.
\newblock Local linear forecasts using cubic smoothing splines.
\newblock \emph{Australian \& New Zealand Journal of Statistics}, 47\penalty0
  (1):\penalty0 87--99, Mar. 2005.
\newblock ISSN 1467-{842X}.
\newblock \doi{10.1111/j.1467-842X.2005.00374.x}.
\newblock URL
  \url{http://onlinelibrary.wiley.com/doi/10.1111/j.1467-842X.2005.00374.x/abstract}.

\bibitem[Kascha(2012)]{Kas11}
C.~J. Kascha.
\newblock A comparison of estimation methods for vector autoregressive
  moving-average models.
\newblock \emph{Econometric Reviews}, 31\penalty0 (3):\penalty0 297--324, 2012.
\newblock ISSN 0747-4938.

\bibitem[Ling and McAleer(2010)]{LinM10}
S.~Ling and M.~McAleer.
\newblock A general asymptotic theory for time-series models.
\newblock \emph{Stat. Neerl.}, 64\penalty0 (1):\penalty0 97--111, 2010.
\newblock ISSN 0039-0402.
\newblock \doi{10.1111/j.1467-9574.2009.00447.x}.
\newblock URL \url{http://dx.doi.org/10.1111/j.1467-9574.2009.00447.x}.

\bibitem[L{\"u}tkepohl(1987)]{Lut_aggregate}
H.~L{\"u}tkepohl.
\newblock \emph{Forecasting aggregated vector {ARMA} processes}, volume 284 of
  \emph{Lecture Notes in Economics and Mathematical Systems}.
\newblock Springer-Verlag, Berlin, 1987.
\newblock ISBN 3-540-17208-4.
\newblock \doi{10.1007/978-3-642-61584-9}.
\newblock URL \url{http://dx.doi.org/10.1007/978-3-642-61584-9}.

\bibitem[L{\"u}tkepohl(2005)]{Lut}
H.~L{\"u}tkepohl.
\newblock \emph{New introduction to multiple time series analysis}.
\newblock Springer-Verlag, Berlin, 2005.
\newblock ISBN 3-540-40172-5.

\bibitem[Makridakis and Hibon(2000)]{MakSH00}
S.~Makridakis and M.~Hibon.
\newblock The {M3}-competition: results, conclusions and implications.
\newblock \emph{International Journal of Forecasting}, 16\penalty0
  (4):\penalty0 451--476, Oct. 2000.
\newblock ISSN 0169-2070.
\newblock \doi{10.1016/S0169-2070(00)00057-1}.
\newblock URL
  \url{http://www.sciencedirect.com/science/article/pii/S0169207000000571}.

\bibitem[Moon et~al.(2012)Moon, Hicks, and Simpson]{MooHS12}
S.~Moon, C.~Hicks, and A.~Simpson.
\newblock The development of a hierarchical forecasting method for predicting
  spare parts demand in the {South} {Korean} {Navy—A} case study.
\newblock \emph{International Journal of Production Economics}, 140\penalty0
  (2):\penalty0 794--802, Dec. 2012.
\newblock ISSN 0925-5273.
\newblock \doi{10.1016/j.ijpe.2012.02.012}.
\newblock URL
  \url{http://www.sciencedirect.com/science/article/pii/S0925527312000709}.

\bibitem[Moon et~al.(2013)Moon, Simpson, and Hicks]{MooHS13}
S.~Moon, A.~Simpson, and C.~Hicks.
\newblock The development of a classification model for predicting the
  performance of forecasting methods for naval spare parts demand.
\newblock \emph{International Journal of Production Economics}, 143\penalty0
  (2):\penalty0 449--454, June 2013.
\newblock ISSN 0925-5273.
\newblock \doi{10.1016/j.ijpe.2012.02.016}.
\newblock URL
  \url{http://www.sciencedirect.com/science/article/pii/S0925527312000746}.

\bibitem[Muth(1960)]{Muth}
J.~F. Muth.
\newblock {Optimal Properties of Exponentially Weighted Forecasts}.
\newblock \emph{Journal of the American Statistical Association}, 55\penalty0
  (290):\penalty0 299--306, June 1960.
\newblock ISSN 01621459.
\newblock \doi{10.2307/2281742}.
\newblock URL \url{http://dx.doi.org/10.2307/2281742}.

\bibitem[Peligrad and Utev(2006)]{PelU06}
M.~Peligrad and S.~Utev.
\newblock Central limit theorem for stationary linear processes.
\newblock \emph{Ann. Probab.}, 34\penalty0 (4):\penalty0 1608--1622, 2006.
\newblock ISSN 0091-1798.
\newblock \doi{10.1214/009117906000000179}.
\newblock URL \url{http://dx.doi.org/10.1214/009117906000000179}.

\bibitem[Sbrana and Silvestrini(2013)]{SS}
G.~Sbrana and A.~Silvestrini.
\newblock Forecasting aggregate demand: Analytical comparison of top-down and
  bottom-up approaches in a multivariate exponential smoothing framework.
\newblock \emph{International Journal of Production Economics}, 146\penalty0
  (1):\penalty0 185--198, Nov. 2013.
\newblock ISSN 0925-5273.
\newblock \doi{10.1016/j.ijpe.2013.06.022}.
\newblock URL
  \url{http://www.sciencedirect.com/science/article/pii/S0925527313002922}.

\bibitem[Weatherford et~al.(2001)Weatherford, Kimes, and Scott]{WeaKS01}
L.~R. Weatherford, S.~E. Kimes, and D.~A. Scott.
\newblock Forecasting for hotel revenue management: Testing aggregation against
  disaggregation.
\newblock \emph{The Cornell Hotel and Restaurant Administration Quarterly},
  42\penalty0 (4):\penalty0 53--64, Aug. 2001.
\newblock ISSN 0010-8804.
\newblock \doi{10.1016/S0010-8804(01)80045-8}.
\newblock URL
  \url{http://www.sciencedirect.com/science/article/pii/S0010880401800458}.

\bibitem[Widiarta et~al.(2009)Widiarta, Viswanathan, and Piplani]{WIDIARTAetal}
H.~Widiarta, S.~Viswanathan, and R.~Piplani.
\newblock Forecasting aggregate demand: An analytical evaluation of top-down
  versus bottom-up forecasting in a production planning framework.
\newblock \emph{International Journal of Production Economics}, 118\penalty0
  (1):\penalty0 87--94, March 2009.
\newblock URL \url{http://ideas.repec.org/a/eee/proeco/v118y2009i1p87-94.html}.

\bibitem[Zotteri and Kalchschmidt(2007)]{ZotK07}
G.~Zotteri and M.~Kalchschmidt.
\newblock Forecasting practices: Empirical evidence and a framework for
  research.
\newblock \emph{International Journal of Production Economics}, 108\penalty0
  (1–2):\penalty0 84--99, July 2007.
\newblock ISSN 0925-5273.
\newblock \doi{10.1016/j.ijpe.2006.12.004}.
\newblock URL
  \url{http://www.sciencedirect.com/science/article/pii/S0925527306003069}.

\bibitem[Zotteri et~al.(2005)Zotteri, Kalchschmidt, and Caniato]{ZotKC05}
G.~Zotteri, M.~Kalchschmidt, and F.~Caniato.
\newblock The impact of aggregation level on forecasting performance.
\newblock \emph{International Journal of Production Economics},
  93–94:\penalty0 479--491, Jan. 2005.
\newblock ISSN 0925-5273.
\newblock \doi{10.1016/j.ijpe.2004.06.044}.
\newblock URL
  \url{http://www.sciencedirect.com/science/article/pii/S092552730400266X}.

\end{thebibliography}

\end{document}